\documentclass{article}

\usepackage{booktabs} % For formal tables

\usepackage{amsmath}
\usepackage{amsthm}
\usepackage{amsxtra}
\usepackage{amssymb}

\usepackage{verbatim} % for multiline comments
\usepackage{hyperref}
\usepackage[utf8]{inputenc}

\usepackage{graphicx}
\usepackage[noadjust]{cite}
\usepackage{xspace}

\usepackage[usenames]{xcolor}

\usepackage[compactparagraph]{pauli_new}

\usepackage{tikz}
\usetikzlibrary{arrows,shapes.geometric,matrix}
\usepackage{pgfplots}
\pgfplotsset{compat=1.14}

\usepackage{authblk}

\newcommand{\probname}[1]{\textsc{#1}}
\newcommand{\setname}[1]{\textit{#1}}
\newcommand{\graphtypename}[1]{\textup{\textbf{#1}}}

\renewcommand{\patterns}[1]{\setname{patterns}(#1)}
\newcommand{\transactions}[1]{\setname{transactions}(#1)}
\newcommand{\labelop}[1]{\setname{label}(#1)}

\newcommand{\enum}[1]{#1\probname{\nobreakdash-enumerate}}
\newcommand{\extend}[1]{#1\probname{\nobreakdash-extend}}
\newcommand{\extendible}[1]{#1\probname{\nobreakdash-extendible}}
\newcommand{\extendiblek}[2]{#1\probname{\nobreakdash-extendible}\langle #2 \rangle}

\newcommand{\Max}[1]{\probname{Max}(#1)}
\newcommand{\MaxFIS}{\probname{MaxFIS}}
\newcommand{\MaxFFIS}{\probname{MaxFFIS}}
\newcommand{\MaxFS}[1]{\probname{MaxFS}(#1)}
\newcommand{\MaxFStau}[1]{\probname{MaxFS}^{\tau}(#1)}
\newcommand{\MaxFSQSn}{\probname{MaxFSQS}}
\newcommand{\MaxFFS}[1]{\probname{MaxFFS}(#1)}

\newcommand{\MaxSQSn}{\probname{MaxSQS}}

\newcommand{\Graphs}{\graphtypename{G}}
\newcommand{\BDG}[1]{\graphtypename{BDG}^{#1}}
\newcommand{\BTW}[1]{\graphtypename{BTW}^{#1}}
\newcommand{\DAG}{\graphtypename{DAG}}
\newcommand{\DirG}{\graphtypename{DirG}}
\newcommand{\T}{\graphtypename{T}}
\newcommand{\PLN}{\graphtypename{PLN}}

\newcommand{\ffpp}{\probname{ffbp}\xspace}
\newcommand{\ffp}{\probname{ffp}\xspace}
\newcommand{\G}{\mathcal{G}}
\renewcommand{\P}{\mathcal{P}}
\newcommand{\Q}{\mathcal{Q}}
\renewcommand{\R}{\mathcal{R}}
\renewcommand{\L}{\mathcal{L}}
\newcommand{\sharpP}{\probname{\#P}}
\newcommand{\sharpR}{\#\R}

\newcommand{\supp}{\operatorname{supp}}

\newcommand{\codeURL}{\url{https://people.mpi-inf.mpg.de/~pmiettin/frequency-based-reductions/}}

\makeatletter
\newif\ifshortversion\shortversiontrue
\@ifclassloaded{IEEEtran}{\shortversiontrue}{\shortversionfalse}
\newcommand{\ifICDM}[2]{\ifshortversion #1\else #2\fi}
\makeatother

\title{Reductions for Frequency-Based Data Mining Problems}
\author[1]{Stefan Neumann\thanks{The first author gratefully acknowledges the
	financial support from the Doctoral Programme ``Vienna Graduate School on
		Computational Optimization'' which is funded by the Austrian Science Fund
		(FWF, project no.~W1260-N35).}}
\author[2]{Pauli Miettinen}
\affil[1]{University of Vienna,
			Faculty of Computer Science,
			Vienna, Austria
			\texttt{stefan.neumann@univie.ac.at}}
\affil[2]{Max Planck Institute for Informatics,
	Saarland Informatics Campus, Germany
	\texttt{pauli.miettinen@mpi-inf.mpg.de}}

\date{}

\begin{document}

\maketitle

\begin{abstract}
Studying the computational complexity of problems is one of the -- if not the
-- fundamental questions in computer science. Yet, surprisingly little is known
about the computational complexity of many central problems in data mining. In
this paper we study frequency-based problems and propose a new type of
reduction that allows us to compare the complexities of the maximal frequent
pattern mining problems in different domains (e.g. graphs or sequences).
Our results extend those of Kimelfeld and Kolaitis [ACM TODS, 2014]
to a broader range of data mining problems.
Our results show that, by allowing constraints in the pattern space, the complexities
of many maximal frequent pattern mining problems collapse.
These problems include maximal frequent subgraphs in labelled graphs,
maximal frequent itemsets, and maximal frequent subsequences with no
repetitions.
In addition to theoretical interest, our
results might yield more efficient algorithms for the studied problems.

%%% Local Variables:
%%% mode: latex
%%% TeX-master: "icmd"
%%% End:

\end{abstract}

\section{Introduction}

Computational complexity is a fundamental concept in computer science, with the
\Poly{} vs.\ \NP{} question being the most famous open problem in the field.
Yet, outside some \NP- and \sharpP-hardness proofs, computational complexity of the
central data mining problems is surprisingly little studied. This is perhaps
even more true for the \emph{frequency-based problems}, that is, for problems
where the goal is to enumerate all sufficiently frequent patterns (that admit
other possible constraints). Problems such as frequent itemset mining,
frequent subgraph mining, and frequent subsequence mining all belong to
this family of problems. Often the only computational complexity argument
for these problems is the observation that the output can be exponentially
large with respect to the input, and hence any algorithm might need
exponential time to enumerate the results.

We argue that this view is too limited for two reasons. First, there are more
fine-grained models of complexity than just the running time. In particular, for
enumeration problems we can use the framework of Johnson et
al.~\cite{johnson88generating}: in short, instead of studying the total running
time with respect to the input size, we can consider it as a function of the
total size of input \emph{and output}, or study the time it takes to create a
\emph{new} pattern when a set of patterns is already known (see
Section~\ref{sec:enumeration-problems} for more details). This
framework allows us to argue about the time complexity of enumeration problems with
potentially exponential output sizes. Another approach is the counting complexity
framework of Valiant~\cite{valiant79complexity} (see Section~\ref{sec:counting-complexity}).

The second reason why we argue that the ``output is exponential'' is a too limited view
for the computational complexity is that a significant question in computational
complexity is the relationships between the problems, that is, questions like
``can we solve problem $X$ efficiently if we can solve problem $Y$
efficiently?'' The main tool for answering these kinds of questions are 
\emph{reductions} between problems. In this work, we introduce a new type of
reduction between frequency-based problems called \emph{maximality-preserving
  reduction} (see Section~\ref{Sec:frequency-reductions}). Our reduction maps
the maximal patterns of one problem to the maximal patterns of the other
problem, thus allowing us to study questions like ``can we find the maximal
frequent subgraphs on labelled graphs using maximal frequent itemset mining
algorithms?'' Surprisingly, the answer to this question turns out to be
positive, although it requires that we consider specially constrained
maximal frequent pattern mining problems; we call the general class of such problems
\emph{feasible frequency-based problems} (see Section~\ref{Sec:constraining}). 

\paragraph*{Our Contributions}

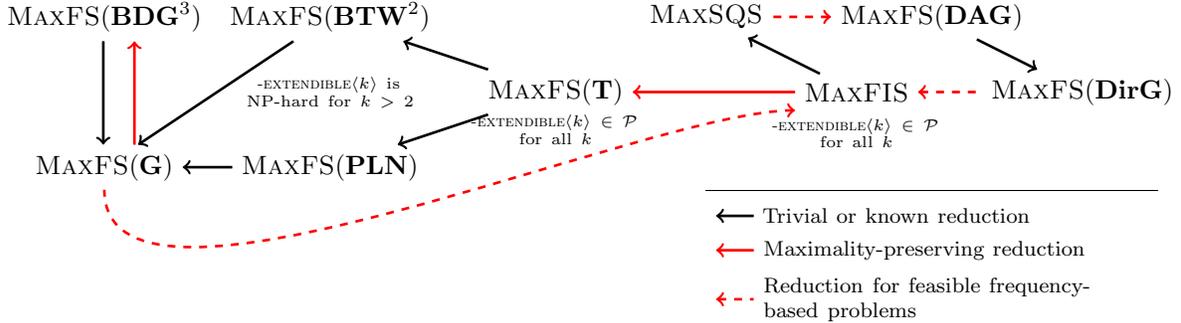
\begin{figure*}
  \begin{centering}
    %% arXiv
    \hspace*{-20mm}
    \tikzstyle{old}=[line width=1pt,draw=black,fill=black,<-]
\tikzstyle{newuc}=[line width=1pt,draw=red,fill=red,<-]
\tikzstyle{newc}=[line width=1pt,draw=red,dashed,<-]
%\tikzstyle{newres}=[star, star points=5, star point ratio=2.25, inner sep=1.3pt, draw=red, fill=red]
\begin{tikzpicture}
  \node (MaxFSG) at (0, 1) {$\MaxFS{\Graphs}$};
  \node (MaxFSBDG) at (0, 3) {$\MaxFS{\BDG{3}}$};
  \draw [old] ([xshift=-.4]MaxFSG.north) -- ([xshift=-.4]MaxFSBDG.south);
  \draw [newuc] ([xshift=.4cm]MaxFSBDG.south) -- ([xshift=.4cm]MaxFSG.north);
  \node (MaxFSPLN) at (3, 1) {$\MaxFS{\PLN}$};
  \draw [old] (MaxFSG) -- (MaxFSPLN);
  \node (BTWext) at (3, 2) [text width=3cm,text centered] {\fontsize{6}{3}\selectfont $\extendiblek{}{k}$ is \NP-hard for $k > 2$\par};
  \node (MaxFSBTW) [above of=BTWext] {$\MaxFS{\BTW{2}}$};
  \draw [old] (MaxFSG) -- (MaxFSBTW);
  \node (MaxFST) at (6, 2) {$\MaxFS{\T}$};
  \node at ([yshift=-.2cm]MaxFST.south) [text width=3cm,text centered] {\fontsize{6}{3}\selectfont $\extendiblek{}{k}\in\P$ for all $k$\par};
  \draw [old] (MaxFSPLN) -- (MaxFST);
  \draw [old] (MaxFSBTW) -- (MaxFST);
  \node (MaxFIS) at (10, 2) {$\MaxFIS$};
  \node (MaxFISext) at ([yshift=-.3cm]MaxFIS.south) [text width=3cm,text centered] {\fontsize{6}{3}\selectfont $\extendiblek{}{k}\in\P$ for all $k$\par};
%  \node at ([xshift=-.5cm,yshift=.2cm]MaxFISext.south east) [newres] {};
  \draw [newuc] (MaxFST) -- (MaxFIS);
  \draw [newc] (MaxFIS.south west) .. controls +(left:2cm) and +(down:2.5cm) .. (MaxFSG);
  \node (MaxSQS) at (8, 3) {$\MaxSQSn$};
  \draw [old] (MaxSQS) -- (MaxFIS);
  \node (MaxFSDAG) at (11, 3) {$\MaxFS{\DAG}$};
  \draw [newc] (MaxFSDAG) -- (MaxSQS);
  \node (MaxFSDirG) at (13, 2) {$\MaxFS{\DirG}$};
  \draw [old] (MaxFSDirG) -- (MaxFSDAG);
  \draw [newc] (MaxFIS) -- (MaxFSDirG);
  \matrix (legend) [column 2/.style={nodes={text width=5cm}}] at (11, -0.3) {
    % \footnotesize
    % \begin{tabular}{@{}p{1cm}p{5cm}@{}}
        %         \multicolumn{2}{c}{\textbf{Legend}} \\
        %         \midrule
        %         \node {\footnotesize \textbf{Legend}}; \\  
    \draw [old] (0,0) -- (.5,0); & \node {\footnotesize Trivial or known reduction}; \\
    \draw [newuc] (0,0) -- (.5,0); & \node {\footnotesize Maximality-preserving reduction}; \\
    \draw [newc] (0,0) -- (.5, 0); & \node {\footnotesize Reduction for feasible frequency-based problems\\}; \\
%    \node at (.25, 0) [newres] {}; & \node {\footnotesize New result implied by our reduction\\}; \\
  };
  \draw (legend.north west) -- (legend.north east);
\end{tikzpicture}

%%% Local Variables:
%%% mode: latex
%%% TeX-master: "icmd"
%%% End:
    \caption{The hierarchy of maximal frequency-based problems with the results from this paper.
      Arrows point from the ``easier'' to the ``harder'' problem.
      See Section~\ref{sec:problems-in-this-paper} for the
      abbreviated problem names used in the picture.
      Maximality-preserving reductions are defined in Section~\ref{Sec:frequency-reductions}
      and feasible frequency-based problems are defined in Section~\ref{Sec:constraining}.}
    \label{fig:hierarchy}
  \end{centering}
\end{figure*}

We study a number of maximal pattern mining problems, including maximal subgraph
mining in labelled graphs (and in more restricted structures), maximal frequent
itemset mining, and maximal subsequence mining with no repetitions (see
Section~\ref{sec:problems-in-this-paper} for definitions of all of these
problems).
We summarize our results in 
Figure~\ref{fig:hierarchy}: the arrows show which problem can be reduced to
which other problem either using non-constraining reductions (black and red lines), or
with possible constraints on the feasible solutions (dashed lines). As can be seen in
Figure~\ref{fig:hierarchy}, all problems can be reduced to each other
(potentially with constraints). Given that the constrained reductions are
transitive (Lemma~\ref{lemma:ffpp-reductions-transitive}), we can state our main
result:

\begin{theorem}[Informal] Maximal subgraph mining in labelled graphs (and in
  more restricted structures), maximal frequent itemset mining, and
  maximal subsequence mining with no repetitions are equally hard problems
  when we are allowed to constrain the pattern space.
\end{theorem}

In some sense, our results unify all existing hardness results for frequency-based problems
by putting them into a general framework using maximality-preserving reductions.
These reductions preserve all interesting theoretical aspects
like \NP- or \sharpP-hardness, but are also restricted enough to maintain the special
properties of the transactions.

In fact, from a practical point of view, our reductions show that if we have an algorithm that can
effectively find, say, the maximal frequent itemsets that admit the constraints
from the reductions, we can use that algorithm to solve maximal frequent
subgraph mining and maximal frequent subsequence mining problems efficiently.
Luckily, as we will see in Section~\ref{sec:algor-exper}, the constrained maximal patterns
are indeed easy to mine in practice.
Alternatively, the reductions can be used to guide how
ideas from algorithms for one set of problems can be transferred to algorithms
for the other set of problems (e.g. from frequent subsequence mining to frequent
subgraph mining or vice versa). 

% From a practical point of view, our results for ffp mining problems imply that
% many algorithms can be used to solve much more diverse problems.  In
% particular, the above reductions shows that any algorithm which allows to
% implement the operation $\land$ when creating candidate patterns, can be used
% for graph mining as well.  An example for such an algorithm is
% slim~\cite{smets12slim}.

%% arXiv vs ICDM
\ifICDM{
  Due to space constraints, some proofs and empirical evaluation are postponed to the complete technical report~\cite{neumann17reductions_arxiv}.
}{%
  % arXiv
  This paper is an extended version of our conference paper~\cite{neumann17reductions}. It contains all omitted proofs, and experimental evaluation.
  }
  
\paragraph*{Outline of the Paper}

We will cover the basic definitions and frameworks used in this paper in
Section~\ref{Sec:preliminaries}, where we will also formally define the problems
we are working with.
Section~\ref{Sec:complexity} presents related work and existing hardness results for
the problems we consider.
We introduce the (unconstrained) maximality-preserving
reductions in Section~\ref{Sec:frequency-reductions}. In particular, the reductions
corresponding to the solid red lines in Figure~\ref{fig:hierarchy} are presented in Section~\ref{Sec:reductions}.
The feasible frequency-based problems, and the
corresponding constrained reductions (dashed red lines in Figure~\ref{fig:hierarchy}) and related results are presented in
Section~\ref{Sec:constraining}.
In Section~\ref{sec:algor-exper} we show that our reductions
can be used in practice and yield efficient algorithms.

%%% Local Variables:
%% mode: latex
%% TeX-master: "icmd"
%% End:

\section{Preliminaries}
\label{Sec:preliminaries}

In this section we quickly cover the basic definitions of frequency-based problems, enumeration problems, and counting complexity. In addition, we present the definitions of the problems we consider in the paper.

\subsection{Frequency-based Problems}
\label{Sec:freq-based-problem}
A \emph{frequency-based} problem $\P$ consists of\footnote{A similar definition was given in Gunopulos et al.~\cite{gunopulos2003discovering}.}:
\begin{itemize}
	\item A set of labels $\L$; for example, $\L = \{1, \dots, n\}$.
	\item A set $\transactions{\P}$ consisting of possible transactions
			over the labels $\L$.
	\item A set $\patterns{\P} \subseteq \transactions{\P}$ of possible patterns
			over the labels $\L$.
	\item A partial order $\sqsubseteq$ over $\transactions{\P}$.
%	\item An operation $\sqcup : \transactions{\P} \times \transactions{\P} \to \transactions{\P}$
%			such that $\transactions{\P}$ and $\patterns{\P}$ are closed under $\sqcup$.
\end{itemize}

Given a frequency-based problem $\P$,
a \emph{database} $D_{\P}$ is a finite multiset of elements from $\transactions{\P}$.
For a database $D_{\P}$ and a \emph{support threshold} $\tau$, a pattern $p \in \patterns{\P}$
is called $\tau$-\emph{frequent} if
\[\supp(p, D_{\P}) := |\{ t \in D_{\P} : p \sqsubseteq t \}| \geq \tau .\]
In other words, a pattern $p$ is frequent if it appears in at least $\tau$ transactions
of the database. When $\tau$ is clear from the context, we will call $p$ only \emph{frequent}.
A pattern $p \in \patterns{\P}$ is a \emph{maximal frequent} pattern if $p$ is
frequent and all patterns $q \in \patterns{\P}$ with $p \sqsubset q$ are not frequent.
Given a database $D_{\P}$, we denote the set of all maximal frequent patterns by
$\Max{D_{\P}, \tau}$, i.e.,
$\Max{D_{\P}, \tau} = \{ p \in \patterns{\P} : p\text{ is a maximal } \tau\text{-frequent pattern in } D_{\P} \}$.

When the parameter $\tau$ is not part of the input but fixed to some integer,
we write $\P^{\tau}$ to denote the resulting problem.

\subsection{Enumeration Problems}
\label{sec:enumeration-problems}
An \emph{enumeration relation $\R$} is a set of strings
$\R = \{ (x,y) \} \subset \{0,1\}^* \times \{0,1\}^*$
such that
\begin{align*}
	\R(x) := \{ y \in \{0,1\}^* : (x,y) \in \R \}
\end{align*}
is finite for every $x$. A string $y \in \R(x)$ is called a \emph{witness} for $x$.
We call $\R$ an \emph{\NP-relation} if (1) there exists a polynomial $p$ such that
$|y| \leq p(|x|)$ for all $(x,y) \in \R$, and
(2) there exists a polynomial-time algorithm deciding if $(x,y) \in \R$ for
any given pair $(x,y)$.

Following~\cite{kimelfeld2014complexity}, we define the following problems for an enumeration relation $\R$:
\begin{itemize}
\item $\enum{\R}$: The input is a string $x$.
	The task is to output the set $\R(x)$ without repetitions.
\item $\extend{\R}$:
	The input is a string $x$ and a set $Y \subseteq \R(x)$.
	The task is to compute a string $y$ such that $y \in \R(x) \setminus Y$
	or to output that no such element exists.
\item $\extendible{\R}$:
	The input is a string $x$ and a set $Y \subseteq \R(x)$.
	The task is to decide whether $\R(x) \setminus Y \neq \emptyset$.
\item $\extendiblek{\R}{k}$: The input is a string $x$ and a set
	$Y \subseteq \R(x)$ with the restriction that $|Y| < k$.
	The task is to decide whether $\R(x) \setminus Y \neq \emptyset$.
\end{itemize}
The problem $\extend{\R}$ can be viewed as the decision version of $\extend{\R}$.
Note that by repeatedly running an algorithm for
$\extend{\R}$, one can solve $\enum{\R}$.
Further observe that any algorithm solving $\extend{\R}$ can be used
to solve $\extendible{\R}$.

\paragraph*{Enumeration Complexity}
Johnson et al.~\cite{johnson88generating} introduced different notions for the
complexity of enumeration problems. Let $\R$ be an enumeration relation. An
algorithm solving $\enum{\R}$ is called an \emph{enumeration algorithm}.

For enumeration problems it might be the case that the output $\R(x)$ is
exponentially larger than the input $x$. Due to this, measuring the
running time of an enumeration algorithm only as a function of $|x|$ can be too
restrictive; instead, one can include the size of $\R(x)$ in the complexity
analysis. Then the running time of an algorithm is measured as function of
$|x| + |\R(x)|$. This consideration gives rise to the following definitions:
\begin{itemize}
	\item An enumeration algorithm runs in \emph{total polynomial time} if its
		running time is polynomial in $|x| + |\R(x)|$. 
	\item An enumeration algorithm has \emph{polynomial delay} if
		the time spent between outputting two consecutive witnesses of $\R(x)$ is always
		polynomial in $|x|$.
	\item An enumeration algorithm runs in \emph{incremental polynomial time} if
		on input $x$ and after outputting a set $Y \subseteq \R(x)$ it takes
		time polynomial in $|x| + |Y|$ to produce the next witness from
		$\R(x) \setminus Y$.
\end{itemize}

We note that $\enum{\R}$ is in incremental polynomial time if and only if $\extend{\R}$ is
in polynomial time. Additionally, observe that a polynomial total time algorithm
can be used to decide if $\R(x) \neq \emptyset$.

\paragraph*{Relationship to Frequency-Based Problems}
We note that frequency-based problems are special cases of enumeration problems.
Let $\P$ be a frequency-based problem.
We define the enumeration relation $\R$ corresponding to $\P$ by setting
\begin{align*}
	\R = \{ (x,y) : x = (D_{\P}, \tau), y \in \Max{ D_{\P}, \tau } \},
\end{align*}
i.e., $\R$ consists of all possible databases $D_{\P}$, support thresholds
$\tau$ and all maximal frequent patterns $y$ for the tuples $(D_{\P}, \tau)$.

Observe that $\R(x) = \R( D_{\P}, \tau ) = \Max{ D_{\P}, \tau }$ and, hence, the
problem $\enum{\R}$ is exactly the same problem as outputting all maximal
frequent patterns in $\Max{ D_{\P}, \tau }$.
The problem $\extend{\R}$ is to output a maximal frequent pattern
in $\Max{ D_{\P}, \tau} \setminus Y$ for a given set of maximal patterns $Y$.
The problems $\extendible{\R}$ and $\extendiblek{\R}{k}$ are the corresponding
decision versions of the problems.

Since $\R$ and $\P$ yield the same enumeration problems, we will also
write $\enum{\P}$, $\extendible{\P}$, $\extend{\P}$, $\extendiblek{\P}{k}$.
Often we will write $\P$ to denote the problem $\enum{\P}$.

 \subsection{Counting Complexity}
 \label{sec:counting-complexity}
 
For a given enumeration relation $\R$, the function $\sharpR : \{0,1\}^* \to \mathbb{N}$
returns the number of witnesses for a given string, i.e., $\sharpR(x) = |\R(x)|$
for $x \in \{0,1\}^*$.
The complexity class $\sharpP$ (pronounced ``sharp P'') contains all functions $\sharpR$ for which $\R$ is an
\NP-relation; it was introduced by Valiant~\cite{valiant79complexity}.
A function $F : \{0,1\}^* \to \mathbb{N}$ is \emph{$\sharpP$-hard} if there exists
a Turing reduction from every function in $\sharpP$ to $F$.

For two \NP-relations $\R, \Q : \{0,1\}^* \to \mathbb{N}$,
a \emph{parsimonious reduction from $\sharpR$ to $\#\Q$} is a polynomial-time computable function 
$f : \{0,1\}^* \to \{0,1\}^*$ such that 
$\sharpR(x) = \#\Q(f(x))$ for all $x \in \{0,1\}^*$.
Note that a parsimonious reduction from a $\sharpP$-hard problem $\R$
to a problem $\Q$ implies that $\Q$ is $\sharpP$-hard.

An example for a $\sharpP$-hard problem is counting \emph{the number} of satisfying
assignments of a SAT formula. Note that such an algorithm can
decide if the given formula is satisfiable or not (by checking if the number of satisfying
assignments is larger than $0$). Hence, $\sharpP$ is a superset of $\NP$.

In fact, Toda and Ogiwara~\cite{toda92counting} showed that all problems in the
polynomial-time hierarchy can be solved in polynomial-time when one has
access to an oracle for a \sharpP-hard function.

Observe that an algorithm solving $\enum{\R}$ can solve $\#\R$ by
counting the number of witnesses in its output.

\subsection{Problems Considered in This Paper}
\label{sec:problems-in-this-paper}
All problems considered in this paper are frequency-based problems.
For the sake of brevity, we only define $\L$, $\transactions{\cdot}$,
$\patterns{\cdot}$, and $\sqsubseteq$ for each problem (see, e.g., \cite{aggarwal15data} for more thorough definitions).

The \emph{maximal frequent itemset mining} problem, denoted as \MaxFIS, is as follows:
We have $n$ labels $\L = \{ 1, \dots, n \}$;
$\transactions{\MaxFIS}$ and $\patterns{\MaxFIS}$ are given by $2^{\L}$;
$\sqsubseteq$ is the standard subset relationship $\subseteq$.

The \emph{maximal frequent subsequence mining} problem, denoted as $\MaxSQSn$, is as follows:
$\L = \{ 1, \dots, n \}$ is the set of labels.
A \emph{sequence} $S = \langle S_1, \dots, S_m \rangle$ of length $m$
consists of $m$ events $S_i$ with $S_i \in \L$;
we require that \emph{each label appears at most once per sequence}.
The sets $\transactions{\MaxSQSn}$ and $\patterns{\MaxSQSn}$ are the sets consisting of all
sequences of arbitrary lengths.
For two sequences $S = \langle S_1, \dots S_r \rangle$ and $T = \langle T_1, \dots, T_k \rangle$,
we have $T \sqsubseteq S$ if $k \leq r$ and there exist indices $1 \leq i_1 \leq \cdots \leq i_k \leq r$
such that $T_j = S_{i_j}$ for each $j = 1, \dots, k$.

Let $\G$ be a class of \emph{vertex-labelled} graphs, which contain
each label at most once.
The \emph{maximal frequent subgraph mining} problem, $\MaxFS{\G}$, is as follows:
We have $n$ labels $\L = \{ 1, \dots, n \}$;
$\transactions{\MaxFS{\G}}$ and $\patterns{\MaxFS{\G}}$ are given by all
labelled graphs in $\G$ with labels from $\L$;
$\sqsubseteq$ is the standard subgraph relationship for labelled graphs (i.e.,
we consider arbitrary subgraphs, not necessarily induced subgraphs).

In the remainder of the paper, we will consider the following graph classes, all of which
are labelled and connected:
\begin{itemize}
\item $\T$ --- undirected trees,
\item $\BDG{b}$ --- undirected graphs of bounded degree at most $b$,
\item $\BTW{w}$ --- undirected graphs of bounded treewidth at most $w$,
\item $\PLN$ --- undirected planar graphs,
\item $\Graphs$ --- general undirected graphs,
\item $\DAG$ --- directed acyclic graphs,
\item $\DirG$ --- directed graphs.
\end{itemize}

Throughout the paper we will only consider labelled graphs \emph{in which
each label appears at most once}. In this restricted setting, the subgraph
isomorphism problem can be solved in polynomial-time.
This a necessary condition for our reductions to work since
Kimelfeld and Kolaitis~\cite[Prop.~3.4]{kimelfeld13complexity} showed that for certain
\emph{unlabelled} graph classes $\G$, $\MaxFS{\G}$ is not an \NP-relation.

\section{Related Work}
\label{Sec:complexity}

\paragraph*{Counting Complexity}
The study of counting problems was initiated when
Valiant~\cite{valiant79complexity} introduced \sharpP.
Provan and Ball~\cite{provan83complexity} showed \sharpP-hardness for many graph
problems such as counting the number of maximal independent sets in bipartite
graphs. Later, more \sharpP-hardness results were obtained for even more restricted
graph classes~\cite{hunt98complexity,vadhan09complexity}.

Johnson et al.~\cite{johnson88generating} introduced the notions of polynomial
total time, polynomial delay, and incremental polynomial time to obtain a better
understanding of the computational complexity of enumeration problems.

\paragraph*{Computational Complexity of Data Mining Problems}
Gunopulos et al.~\cite{gunopulos2003discovering} introduced a general
class of problems similar to frequency-based problems.
For this class of problems, they proved $\sharpP$-hardness for mining frequent sets,
and provided an algorithm to mine maximal frequent sets.

Yang~\cite{yang2004complexity} proved $\sharpP$-hardness for determining
the number of \emph{maximal} frequent itemsets and other problems.
\begin{theorem}[Yang~\cite{yang2004complexity}]
\label{Thm:yang2004complexity}
	The following problems are \sharpP-complete:
	$\MaxFIS$, $\MaxFS{\T}$, $\MaxFS{\Graphs}$, $\MaxSQSn$.
\end{theorem}

Boros et al.~\cite{boros2003maximal} showed that given a set of maximal
frequent itemsets $Y$, it is \NP-complete to decide whether there exists another
maximal frequent itemset that is not contained in $Y$.
\begin{theorem}[Boros et al.~\cite{boros2003maximal}]
	$\extendible{\MaxFIS}$ and $\extend{\MaxFIS}$ are \NP-complete.
\end{theorem}

Kimelfeld and Kolaitis~\cite{kimelfeld13complexity} proved structural results on mining
frequent subgraphs of certain graph classes.
Their results allow to distinguish the computational complexities of
$\MaxFS{\T}$ and $\MaxFS{\G}$ where
$\G$ is either $\Graphs$, $\PLN$, $\BDG{b}$ with $b > 2$, or $\BTW{w}$ with $w > 1$.
This is also depicted in Figure~\ref{fig:hierarchy}.

\begin{theorem}[Kimelfeld and Kolaitis~\cite{kimelfeld2014complexity}]
\label{Thm:kimelfeld2014complexity}
	For every fixed $k$, the problem $\extendiblek{\MaxFS{\T}}{k}$ can
	be solved in polynomial time.

	For every fixed $\tau$, the problem $\enum{\MaxFStau{\G}}$ can be solved in
	polynomial time for any class of graphs $\G$ from Section~\ref{sec:problems-in-this-paper}.

	The following problems are \NP-complete:
	\begin{itemize}
		\item $\extendible{\MaxFS{\G}}$ for 
	   		$\G \in \{ \Graphs, \PLN, \BDG{b}, \BTW{w} \}$
			with $w \geq 1$ and $b \geq 3$.
		\item $\extendiblek{\MaxFS{\G}}{k}$ for 
	   		$\G \in \{ \Graphs, \PLN, \BDG{b}, \BTW{w} \}$
			with $w > 1$ and $b > 2$ and for every $k > 2$.
	\end{itemize}
\end{theorem}

In the journal version~\cite{kimelfeld2014complexity} of their
paper~\cite{kimelfeld13complexity}, Kimelfeld and Kolaitis give
computational hardness results for subgraph mining problems in which the
set \patterns{} is more restricted than \transactions{}. For example, they consider
the computational complexity of mining maximal subtrees from planar graphs.
They also consider mining unlabelled maximal subgraphs.

\paragraph*{Mining Maximal Frequent Patterns}
Many algorithms were proposed to mine maximal frequent patterns from different
types of data such as itemsets~\cite{han04mining,burdick05mafia,bayardo98efficiently},
subsequences~\cite{agrawal95mining},
trees~\cite{zaki05efficiently,xiao03efficient},
and general graphs~\cite{kuramochi01frequent}.
However, the main focus of those papers was not to investigate the computational
complexity of these problems.
See (for example) the book by Aggarwal~\cite{aggarwal15data} for many more references to
algorithms for efficiently computing maximal frequent patterns.

\paragraph*{Constraint-based Pattern Mining}
Many algorithms were proposed to mine frequent patterns with
constraints on the structure of the
patterns~\cite{ng98exploratory,grahne00efficient,bonchi03exante,bonchi04closed,garofalakis99spirit,pei02mining,bonchi09constraint}.
Due to lack of space we cannot
review all of them, but refer to Han et al.~\cite{han07frequent} for references to many
papers on constrained pattern mining.
Greco et al.~\cite{greco08mining} presented techniques for mining taxonomies of
process models which can also be viewed as constraint-based pattern mining.
The work on constraint programming for itemset mining by
Raedt et al.~\cite{raedt08constraint} and follow-up work (e.g.~\cite{guns13pattern})
can also be used to mine itemsets or other frequency-based problems with constraints.

\section{Maximality-Preserving Reductions}
\label{Sec:frequency-reductions}

In this section, we introduce maximality-preserving reductions and state some of
their properties in Section~\ref{sec:definition-and-properties}.
In Section~\ref{Sec:reductions}, we prove reductions between the problems $\MaxFIS$,
$\MaxSQSn$, and $\MaxFS{\G}$ for $\G \in \{ \T, \BDG{3}, \Graphs \}$.
Combining our reductions with the statements from
Section~\ref{Sec:complexity}, we arrive at the following theorem.

\begin{theorem}
Our reductions imply the following hardness results:
\begin{enumerate}
\item For any fixed $k$, $\extendiblek{\MaxFIS}{k}$ can
be solved in polynomial time.

\item For any fixed $\tau$, $\enum{\MaxFIS^{\tau}}$ can
be solved in polynomial time.

\item The problems $\MaxFS{\Graphs}$ and $\MaxFS{\BDG{3}}$ exhibit exactly the same
hardness w.r.t.\ the notions of Sections~\ref{sec:enumeration-problems}
and~\ref{sec:counting-complexity}.
More concretely, let $\P$ be $\MaxFS{\Graphs}$ or $\MaxFS{\BDG{3}}$.
Then the following statements are true:
\begin{itemize}
	\item $\enum{\P}$ is \sharpP-hard.
	\item $\extendible{\P}$ is \NP-hard.
	\item For $k > 2$, the problem $\extendiblek{\P}{k}$ is \NP-hard.
	\item For fixed $\tau$, the problem $\enum{\P^{\tau}}$ is solvable in
		polynomial time.
\end{itemize}
\end{enumerate}
\end{theorem}

The proof of the theorem follows from our reductions later in this section and
the theorems from Section~\ref{Sec:complexity}.

\subsection{Definition and Properties}
\label{sec:definition-and-properties}
We formally define maximality-preserving reductions to make explicit which properties are
required by reductions in order to be useful for understanding the complexity
of frequency-based problems w.r.t.\ to the notions of
Sections~\ref{sec:enumeration-problems}
and~\ref{sec:counting-complexity}.

\begin{definition}
\label{Def:MaxPreservingReduction}
Let $\P$ and $\Q$ be two frequency-based problems, let
$D_{\P}$ be a database for $\P$, and let $\tau$ be a support threshold.
A \emph{maximality-preserving reduction} from $\P$ to $\Q$
defines an instance $(D_{\Q}, \tau)$ using a polynomial-time
computable injective function $f\colon \transactions{\P} \to \transactions{\Q}$
with the following properties:
\begin{enumerate}
	\item \label{item:patternSubsets1}
		$f(\patterns{\P}) \subseteq \patterns{\Q}$.
	\item \label{item:subsetPreserving1}
		For all $p, p' \in \transactions{\P}$,
		$p \sqsubseteq_{\P} p'$ if and only if $f(p) \sqsubseteq_{\Q} f(p')$.
	\item \label{item:efficientInverse1}
		The inverse $f^{-1}\colon \transactions{\Q} \to \transactions{\P}$ of $f$ can be
			computed in polynomial time.
	\item \label{item:maxPreserving1}
		$p \in \Max{D_{\P}, \tau}$ if and only if
	   	$f(p) \in \Max{D_{\Q}, \tau}$, where $D_{\Q} = f(D_{\P}) = \{ f(t) : t \in D_{\P} \}$.
		Additionally, for all $q \in \Max{D_{\Q}, \tau}$ the preimage
		$f^{-1}(q)$ exists.
\end{enumerate}
\end{definition}

Intuitively, the properties can be interpreted as follows:
Property~\ref{item:patternSubsets1} asserts that $f$ maps valid patterns from $\patterns{\P}$
to valid patterns in $\patterns{\Q}$;
this condition is necessary if $\patterns{\Q} \subsetneq \transactions{\Q}$.
Property~\ref{item:subsetPreserving1} asserts that $f$ maintains subset properties.
Property~\ref{item:efficientInverse1} will be necessary to recover patterns in $\P$
from those found in $\Q$.
Property~\ref{item:maxPreserving1} requires that the maximal frequent
patterns in $D_{\P}$ are the same as those in $D_{\Q}$ under the mapping $f$;
here, the database $D_{\Q}$ is given by applying the function $f$ to each
transaction in $D_{\P}$.

\paragraph*{Properties}
Observe that Property~\ref{item:maxPreserving1}
implies that there exists a bijective relationship between the maximal frequent patterns
in $D_{\P}$ and in $D_{\Q}$.
Hence, we have $|\Max{D_{\P}, \tau}| = |\Max{D_{\Q}, \tau}|$.
This shows that maximality-preserving reductions are special cases of parsimonious reductions
and that they preserve $\sharpP$-hardness.

In fact, maximality-preserving reductions are slightly stronger than parsimonious reductions.
They do not only preserve the number of maximal frequent patterns in both databases,
but they enable us to recover the maximal frequent patterns in $D_{\P}$ from those in $D_{\Q}$:
By injectivity of $f$ and due to Property~\ref{item:maxPreserving1}, we can reconstruct
$\Max{D_{\P},\tau}$ in polynomial time from $\Max{D_{\Q},\tau}$.
Hence, maximality-preserving reductions can be used to argue about the complexity of
extendibility problems as discussed in Section~\ref{sec:enumeration-problems}.

Further, note that by choice of $D_{\Q}$ in Property~\ref{item:maxPreserving1},
$D_{\Q}$ has the same number of transactions as $D_{\P}$, and
that no dependency within different transactions is created by the mapping $f$.
Additionally, by Property~\ref{item:subsetPreserving1}, the support of a
pattern $p$ in $D_{\P}$ is a lower bound on the support of $f(p)$ in
$D_{\Q}$ (since for each transaction $t \in D_{\P}$ with $p \sqsubseteq t$,
$f(p) \sqsubseteq f(t)$).

However, although the number of transactions and \emph{maximal} frequent patterns in both databases remains the same,
the number of \emph{frequent} patterns in $D_{\Q}$ might be exponentially larger
than the number of frequent patterns in $D_{\P}$. For example, this is the case in the reduction
in Lemma~\ref{Thm:ReductionBDG3}.

\subsection{Reductions}
\label{Sec:reductions}
In this section, we present three maximality-preserving reductions.
Reductions similar to ones in Lemmas~\ref{Thm:MaxFIS} and~\ref{Thm:ReductionMaxSQS}
were already presented by Yang~\cite{yang2004complexity}, Kimelfeld and
Kolaitis~\cite{kimelfeld2014complexity} and other authors.
We only prove Property~\ref{item:maxPreserving1} of maximality-preserving reductions.
The proofs of Properties~\ref{item:patternSubsets1}--\ref{item:efficientInverse1} are
straight-forward and follow from the definitions of the mapping $f$.

\paragraph*{Reduction from $\MaxFIS$ to $\MaxFS{\T}$}
We show how to mine maximal itemsets by mining maximal subtrees.
\begin{lemma}
\label{Thm:MaxFIS}
There exists a maximality-preserving reduction from $\MaxFIS$ to $\MaxFS{\T}$.
\end{lemma}
\begin{proof}
Consider $\MaxFIS$ with labels $\L = \{1, \dots, n\}$.
We construct trees over labels from the alphabet $\L' = \{r, 1, \dots, n\}$,
where $r$ is the label of the root nodes in the trees. For simplicity, we do not
distinguish between vertices and their labels.

\emph{Construction of $f$.}
An itemset $\{i_1, \dots, i_k\} \in \transactions{\MaxFIS}$ is mapped to a tree
of depth 1 with root $r$ and children $i_1, \dots, i_k$, i.e., the tree has an
edge $(r, i_j)$ for all $j = 1, \dots, k$.

\emph{Maximality-preserving.}
Observe that there exists a bijection between itemsets $I \subseteq \L$ and trees $f(I)$.
Further note that for two itemsets $I$ and $J$, $I \subseteq J$ if and only if $f(I) \subseteq f(J)$.
It follows that an itemset $I$ and a tree $f(I)$ must have the same supports in
$D_{\MaxFIS}$ and in $D_{\MaxFS{\T}}$, respectively. The maximality then follows from the
subset-property we observed.
\end{proof}

\paragraph*{From $\MaxFIS$ to $\MaxSQSn$}
We show how to mine maximal itemsets by mining maximal subsequences.
\begin{lemma}
\label{Thm:ReductionMaxSQS}
There exists a maximality-preserving reduction from $\MaxFIS$
to $\MaxSQSn$.
\end{lemma}
\begin{proof}
\emph{Construction of $f$.}
Consider $\MaxFIS$ with labels $\L = \{1, \dots, n\}$ and assume the labels are ordered
w.r.t.\ to some arbitrary, but fixed, order $\prec$.
Let $I = \{i_1, \dots, i_m\} \subseteq \L$ be any itemset with $m$ items.
Assume w.l.o.g.\ that the items in $I$ are ordered w.r.t.\ the fixed order, i.e., $i_j \prec i_{j+1}$.
Then $I$ is mapped to the sequence $\langle i_1, \dots, i_m \rangle$ of length $m$.

\emph{Maximality-preserving.}
Observe that there exists a bijection between itemsets $I \subseteq \L$ and sequences $f(I)$
(under the fixed order).
Further observe that for two itemsets $I$ and $J$, $I \subseteq J$ if and only if $f(I) \sqsubseteq f(J)$.
It follows that an itemset $I$ and a sequence $f(I)$ must have the same supports in
$D_{\MaxFIS}$ and in $D_{\MaxSQSn}$, respectively. The maximality then follows from the
subset-property we observed.
\end{proof}

\paragraph*{From $\MaxFS{\Graphs}$ to $\MaxFS{\BDG{3}}$}
We show that mining maximal frequent subgraphs in graphs with degrees bounded by $3$ can
be used to mine maximal frequent subgraphs in general undirected graphs.
Note that this is the tightest result we could hope for, since graphs with degree
bounded by $2$ are simply cycles or line graphs.

\begin{lemma}
\label{Thm:ReductionBDG3}
There exists a maximality-preserving reduction from $\MaxFS{\Graphs}$ to $\MaxFS{\BDG{3}}$.
\end{lemma}
\begin{proof}
\emph{Construction of $f$.}
Let $G = (V,E)$ be a graph with unbounded degree of the vertices
over labels $\L = \{1, \dots, n\}$.
Denote the label of a vertex $v \in V$ by $\labelop{v}$.
We construct a graph $G' = (V', E')$ with bounded degree~$3$
over the set of labels $\L' = \{1, \dots, n \}^2$.

Intuitively, the construction of $f$ is picked such that every original vertex
$v \in V$ is split into a line graph consisting of $n$ vertices $v_i$, where each $v_i$
has an additional non-line-graph-edge in $G'$ iff vertices $v$ and $i$ share an edge in $G$.

Formally, for each vertex $v \in V$, we insert vertices $v_1, \dots, v_n$ into $V'$
with edges $(v_i, v_{i+1})$ for $i = 1, \dots, n-1$.
Each vertex $v_i$ is labeled by $(\labelop{v}, i)$.
For each edge $(u,v) \in E$, we insert an edge $(u_{\labelop{v}}, v_{\labelop{u}})$ into
$G'$.

Observe that the resulting graph $G' = f(G)$ indeed has bounded degree $3$:
Consider any vertex $v_i \in V'$.
The vertex has at most $2$ neighbors from the line graph $(v_1, \dots, v_n)$.
The only additional edge it could have is to vertex $i_{\labelop{v}}$.

\emph{Maximality-preserving.}
Let $p \in \Max{D_{\MaxFS{\Graphs}}, \tau}$.
We need to show that $f(p) \in \Max{D_{\MaxFS{\BDG{3}}}, \tau}$.
By construction of $f$, we have that
$\supp(f(p), D_{\MaxFS{\BDG{3}}}) = \supp(p, D_{\MaxFS{\Graphs}})$;
hence, $f(p)$ is frequent in $D_{\MaxFS{\BDG{3}}}$.
We need to show that $f(p)$ is also maximal. For the sake of contradiction,
suppose there exists a maximal frequent pattern $q$ with $f(p) \sqsubset q$ in $D_{\MaxFS{\BDG{3}}}$.
Then $q$ must contain an edge $(u_i, v_j)$ with $i = \labelop{v}$, $j = \labelop{u}$,
which is not contained in $f(p)$.

\emph{Case 1:}
$u_i \in f(p)$ and $v_j \in f(p)$. Consider the graph $q' =  f(p) \cup (u_i, v_j)$.
Then $f^{-1}( q' )$ exists and must be
frequent in $D_{\MaxFS{\Graphs}}$ by Property~\ref{item:subsetPreserving1}.
This contradicts the maximality of $p$.

\emph{Case 2:}
W.l.o.g.\ assume that $u_i \in f(p)$ and $v_j \not\in f(p)$.
Then since $q$ is maximal and by construction of $f$ and $D_{\MaxFS{\BDG{3}}}$,
$q$ must contain the line graph $L$ with vertices $v_1, \dots, v_n$.
Consider the graph $q' = f(p) \cup (u_i,v_j) \cup L$.
Again by construction of $f$ and $D_{\MaxFS{\BDG{3}}}$,
$q'$ has a preimage $p' = f^{-1}(q')$ which is frequent and satisfies $p \sqsubset p'$.
This is a contradiction to the maximality of $p$.

\emph{Case 3:}
$u_i \not\in f(p)$ and $v_j \not\in f(p)$.
Since $q$ is connected and $f(p) \sqsubset q$, we only need to consider the first
two cases.

Observe that the second part of Property~\ref{item:maxPreserving1} is implied by
the previous three case distinctions.
Proving that $f(p) \in \Max{D_{\MaxFS{\BDG{3}}}, \tau}$ implies
$p \in \Max{D_{\MaxFS{\Graphs}}, \tau}$ can be done similarly to above.
\end{proof}

\section{Constraining the Set of Patterns}
\label{Sec:constraining}

In this section, we generalize frequency-based problems by allowing to constrain
the set of patterns using a feasibility function.
We introduce maximality-preserving reductions for this class of problems and
prove that all problems discussed in this paper exhibit exactly the same
hardness after introducing the feasibility function.

\subsection{Feasible Frequency-Based Problems}
A \emph{feasible frequency-based problem} (\ffpp) $\P$ is
a frequency-based problem with an additional polynomial-time computable operation
$\phi\colon \patterns{\P} \to \{0,1\}$ which can be described using
constant space.
Note that the operation $\phi$ is part of the input for the problem;
this is the reason for restricting the description length of the function
to constant size (otherwise, the description length of the function might be
larger than the database for the problem).
We call $\phi$ the \emph{feasibility function}.

Given a feasible frequency-based problem $\P$,
a pattern $p \in \patterns{\P}$ is a \emph{feasible frequent pattern} (\ffp) if
$p$ is frequent and $\phi(p) = 1$.
The goal is to find all maximal \ffp s; we denote the set of all \ffp s by
$\Max{D_\P, \tau, \phi_{\P}}$.
We define $\MaxFFIS$, $\MaxFSQSn$, and $\MaxFFS{\G}$
for a graph class $\G$ as before for maximal frequency-based problems.

Note that \ffpp s are generalizations of frequency-based problems since setting
$\phi_{\P}$ to the function which is identical to $1$, we obtain the underlying frequency-based problem.

The main result of this section is given in the following theorem.
\begin{theorem}
\label{Thm:ffpp}
The \ffpp-version of all problems discussed in this paper exhibit exactly the same
hardness w.r.t.\ the notions of Sections~\ref{sec:enumeration-problems}
and~\ref{sec:counting-complexity}.
More concretely, let $\P$ be any \ffpp-problem discussed in this paper. Then the
following statements are true:
\begin{itemize}
	\item $\enum{\P}$ is \sharpP-hard.
	\item $\extendible{\P}$ is \NP-hard.
	\item For $k > 2$, the problem $\extendiblek{\P}{k}$ is \NP-hard.
	\item For fixed $\tau$, the problem $\enum{\P^{\tau}}$ is solvable in
		polynomial time.
\end{itemize}
\end{theorem}

Theorem~\ref{Thm:ffpp} shows that the hierarchy given in
Figure~\ref{fig:hierarchy} for frequency-based problems completely collapses
when a feasibility function is introduced to the problem.
Note that many practical algorithms (like the Apriori algorithm) for finding maximal frequent
patterns allow to add such a feasibility function.
Hence, our reductions give a theoretical justification why many of these algorithms
can be extended to a broader range of problems.

The proof of the theorem follows from the reductions presented later in this
section and the theorems from Section~\ref{Sec:complexity}.

\subsection{Maximality-Preserving Reductions for FFPPs}
We start by defining maximality-preserving reductions between two \ffpp s $\P$ and $\Q$.
\begin{definition}
\label{Def:FFPPReduction}
	Let $\P$ and $\Q$ be two \ffpp s.
	Let $D_{\P}$ be a database for $\P$, let $\phi_{\P}$ be
	the feasibility function for $\P$, and let $\tau$ be a support threshold.

	A \emph{maximality-preserving reduction} from $\P$ to $\Q$
	defines an instance $(D_{\Q}, \tau, \phi_{\Q})$ using a polynomial-time
	computable injective function $f\colon \transactions{\P} \to \transactions{\Q}$
	with the following properties:
	\begin{enumerate}
		\item \label{item:patternSubsets2} $f(\patterns{\P}) \subseteq \patterns{\Q}$.
		\item \label{item:subsetPreserving2}
				For all $p, p' \in \transactions{\P}$, $p \sqsubseteq_{\P} p'$ if and only if
				$f(p) \sqsubseteq_{\Q} f(p')$.
		\item \label{item:efficientInverse2}
				The inverse $f^{-1}\colon \transactions{\Q} \to \transactions{\P}$ of $f$ can be
				computed in polynomial time.
		\item \label{item:maxPreserving2}
			$p \in \Max{D_{\P}, \tau, \phi_{\P}}$ if and only if
			$f(p) \in \Max{D_{\Q}, \tau, \phi_{\Q}}$, where $D_{\Q} = f(D_{\P}) = \{ f(t) : t \in D_{\P} \}$.
			Additionally, for all $q \in \Max{D_{\Q}, \tau, \phi_{\Q}}$ the preimage
			$f^{-1}(q)$ exists.
	\end{enumerate}
\end{definition}

Note that compared to Definition~\ref{Def:MaxPreservingReduction}, we only had
to change Property~\ref{item:maxPreserving2} to assert that the maximal patterns
are feasible.
Further observe that in general the function $\phi_{\Q} = \phi_{\Q}(\phi_{\P},f,f^{-1})$ constructed in the reduction
will depend on $\phi_{\P}$, $f$ and $f^{-1}$.

\paragraph*{Properties}
The rest of this subsection is devoted to proving properties of
maximality-preserving reductions for \ffpp s. First, we show that
maximality-preserving reductions are transitive, which is the crucial property
to argue that one can use multiple reductions in a row.
Second, we show that maximality-preserving reductions for frequency-based
problems imply maximality-preserving reductions for \ffpp s.

The following lemma shows that maximality-preserving reductions for \ffpp s are transitive.
The main challenge will be the construction of the feasibility function.
\begin{lemma}
  \label{lemma:ffpp-reductions-transitive}
  Let $\P, \Q, \R$ be \ffpp s.
  Assume there exist maximality-preserving reductions
  from $\P$ to $\Q$ via a function $g$ and $\phi_{\Q}$, and
  from $\Q$ to $\R$ via a function $h$ and $\phi_{\R}$.
  Then there exists a maximality-preserving reduction from $\P$ to $\R$.
\end{lemma}

\begin{proof}
	Let $D_{\P}$ and $\phi_{\P}$ be an instance for $\P$.
	We construct an instance $(D^*, \phi_*)$ for $\R$:
	We set $f\colon \transactions{\P} \to \transactions{\R}$ to $f(p) = h(g(p))$
	for $p \in \transactions{\P}$.
	For a pattern $r \in \patterns{\R}$, we set
	$\phi_*(r) = 1$ if and only if
	the following four conditions are satisfied:
	(1) $h^{-1}(r)$ and $f^{-1}(r)$ exist;
	(2) $\phi_{\R}(r) = 1$;
	(3) $\phi_{\Q}(h^{-1}(r)) = 1$; and
	(4) $\phi_{\P}(f^{-1}(r)) = 1$.

	We check the properties from Definition~\ref{Def:FFPPReduction}.
	Property~\ref{item:patternSubsets2} and Property~\ref{item:subsetPreserving2}
	are satisfied since $f$ is the composition $g$ and $h$.
	Property~\ref{item:efficientInverse2} holds since $f^{-1} = g^{-1} \circ h^{-1}$
	and both $g^{-1}$ and $h^{-1}$ can be computed in polynomial time.

	The rest of the proof is devoted to proving Property~\ref{item:maxPreserving2}.

	Let $p \in \Max{D_{\P},\tau,\phi_{\P}}$.
	Then $p$ is feasible w.r.t.\ $\phi_{\P}$.
	By the reduction from $\P$ to $\Q$,
	$g(p) \in \Max{D_{\Q},\tau,\phi_{\Q}}$, where
	$D_{\Q} = g( D_{\P} )$. Note that $g(p)$ is feasible w.r.t.\ $\phi_{\Q}$.
	Using the reduction from $\Q$ to $\R$, we obtain
	$r := h(g(p)) \in \Max{ D_{\R}, \tau, \phi_{\R} }$, where
	$D_{\R} = h( D_{\Q} )$; additionally, $r$ is feasible w.r.t.\ $\phi_{\R}$.
	Now observe that $r = f(p)$ and that $r$ is feasible w.r.t.\ the operation
	$\phi_*$ defined above.
	Note that $r$ is frequent in $D^*$ since for each transaction $t \in D_{\P}$
	with $p \sqsubseteq_{\P} t$, $r = f(p) \sqsubseteq_{\R} f(t)$
	by Property~\ref{item:subsetPreserving2} of $f$.
	To prove that $r \in \Max{D^*, \tau, \phi_*}$, it remains to show that
	$r$ is maximal.
	Suppose not. Then there exists a pattern $r' \in \Max{D^*, \tau, \phi_*}$ such that
	$r \sqsubset_{\R} r'$. Since $r'$ is feasible, let $p' = f^{-1}(r')$.
	By Property~\ref{item:subsetPreserving2} of $f$, we have that $p \sqsubset_{\P} p'$
	and that $p'$ is frequent
	since $p' \sqsubset_{\P} t$ for $t \in D_{\P}$ if and only if $f(p') = r' \sqsubset_{\R} f(t)$.
	This contradicts the maximality of $p$.
	Hence, we proved that $r \in \Max{D^*, \tau, \phi_*}$.

	Let $r \in \Max{D^*, \tau, \phi_*}$.
	Since $r$ is feasible w.r.t.\ $\phi_*$, there exists
	$p = f^{-1}(r) \in \patterns{\P}$ that is feasible w.r.t.\ $\phi_{\P}$.
	By Property~\ref{item:subsetPreserving2}, $p$ is frequent in $D_{\P}$.
	It remains to show that $p$ is maximal.
	We argue by contradiction.
	Suppose there exists a frequent pattern $p'$ with $p \sqsubset p'$.
	Then $f(p') \in \Max{D_*,\tau,\phi_*}$ by the previous paragraph,
	and $r \sqsubset f(p')$ by Property~\ref{item:subsetPreserving2} of $f$.
	This contradicts the maximality of $r$.
	Hence, $p \in \Max{ D_{\P}, \tau, \phi_{\P}}$.
\end{proof}

The next lemma shows that if for two frequency-based problems $\P$ and $\Q$
there exists a maximality-preserving reduction from $\P$ to $\Q$,
then there also exists a reduction between the \ffpp-version of these problems.
\begin{lemma}
Let $\P$ and $\Q$ be two frequency-based problems, and let $\P'$ and $\Q'$
be the \ffpp-versions of those problems.
Suppose there exists a maximality-preserving reduction from $\P$ to $\Q$ via
a mapping $g$.

Then there exists a maximality-preserving reduction from $\P'$ to $\Q'$.
\end{lemma}
\begin{proof}
\emph{Construction of $f$.}
We set $f \equiv g$.
Given a pattern $q \in \patterns{\Q}$,
we set $\phi_{\Q'}(q) = 1$ iff
$f^{-1}(q)$ exists and $\phi_{\P'}(f^{-1}(q)) = 1$.

\emph{Maximality-preserving.}
Note that Properties~\ref{item:patternSubsets2}--\ref{item:efficientInverse2}
of maximality-preserving reductions for $f$ are
satisfied since they are satisfied for $g$.
We prove Property~\ref{item:maxPreserving2} of $f$.

Let $p \in \Max{ D_{\P}, \tau, \phi_{\P} }$.
We show that $f(p) \in \Max{ D_{\Q}, \tau, \phi_{\Q}}$.
Observe that $f(p)$ is feasible w.r.t.\ $\phi_{\Q}$ since
$f^{-1}(f(p)) = p$ is feasible w.r.t.\ $\phi_{\P}$.
Note that $f(p)$ is frequent in $D_{\Q}$ by
Property~\ref{item:subsetPreserving2} of $f$.
We need to argue that $f(p)$ is also maximal.
Suppose this is not the case. Then there exists a pattern
$q \in \Max{ D_{\Q}, \tau, \phi_{\Q}}$ such that $f(p) \sqsubset q$.
Since $q$ is feasible, there exists a feasible pattern
$p' = f^{-1}(q) \in \patterns{\P}$. By Property~\ref{item:subsetPreserving2},
we have $p \sqsubset p'$.
Additionally, the pattern $p'$ is frequent in $D_{\P}$:
for each transaction $t \in D_{\Q}$ with $q \sqsubset_{\Q} t$, $p' \sqsubset_{\P} f^{-1}(t)$
(by Property~\ref{item:subsetPreserving2} of $f$ and definition of
 $D_{\Q}$).
This contradicts the maximality of $p$.

Let $q \in \Max{ D_{\Q}, \tau, \phi_{\Q} }$.
Since $q$ is feasible, $p = f^{-1}(q)$ exists
and is feasible w.r.t.\ $\phi_{\P}$.
We show that $p \in \Max{ D_{\P}, \tau, \phi_{\P}}$.
Note that $p$ is frequent in $D_{\P}$ by Property~\ref{item:subsetPreserving2}
of $f$. We prove the maximality of $p$ by contradiction. Suppose there exists a
pattern $p' \in \Max{ D_{\P}, \tau, \phi_{\P}}$ with $p \sqsubset p'$.
Then by the previous paragraph the pattern $f(p')$ is a feasible frequent
pattern in $D_{\Q}$ with $q = f(p) \sqsubset f(p')$. This contradicts the maximality
of $q$.
\end{proof}

\subsection{Reductions}

\paragraph*{From graphs to feasible frequent itemsets}
We show that any algorithm solving the $\MaxFFIS$-problem can be used
to mine maximal frequent subgraphs in general graphs.

\begin{lemma}
\label{Thm:MaxFSGraphsFFIS}
There exists a maximality-preserving reduction from $\MaxFFS{\Graphs}$
to $\MaxFFIS$.
\end{lemma}
\begin{proof}
Let $D_{\MaxFFS{\Graphs}}$ be a database consisting of
labelled graphs from $\Graphs$ with labels from $\{1, \dots, n\}$,
let $\tau$ be a support threshold,
let $\phi_{\MaxFFS{\Graphs}}$ be a feasibility function.

\emph{Construction of $f$.}
For $\MaxFFIS$ we use the labels $\L = \{ 1, \dots, n\}^2$.
Let $G = (V,E)$ be a graph from $D_{\MaxFFS{\Graphs}}$.
We construct an itemset $I(G) := f(G)$ by mapping the graph
onto the labels of its edges, i.e., we construct an
itemset $I(G) = \{ (\labelop{u}, \labelop{v}) : (u,v) \in E \}$.

Given an itemset $I \in \patterns{\MaxFFIS}$, we set
$\phi_{\MaxFFIS}(I) = 1$ iff
(1) $f^{-1}(I)$ exists and $\phi_{\MaxFFS{\Graphs}}(f^{-1}(I)) = 1$, and
(2) for each pair of tuples $(a,b), (c,d) \in I$
	there exists a sequence $(a,b) = (e_1,e_1'), \dots, (e_k,e_k') = (c,d)$
	of tuples $(e_i, e_i') \in I$ with the following property:
	For each pair of consecutive tuples $(e_i, e_i')$ and $(e_{i+1},e_{i+1}')$,
	there exists some $\ell \in \{1,\dots,n\}$ with $\ell \in \{e_i, e_i'\}$ and $\ell \in
	\{e_{i+1},e_{i+1}'\}$.
Intuitively, condition (2) of $\phi_{\MaxFFIS}$ asserts that the graphs corresponding
to the itemset $I$ must be connected.

\emph{Maximality-preserving.}
Note that any feasible frequent itemset in $D_{\MaxFFIS}$ corresponds to a
frequent \emph{connected} graph in $D_{\MaxFFS{\Graphs}}$ due to the choice of
$\phi_{\MaxFFIS}$.
Observe that there exists a bijection between connected subgraphs $G$ and feasible itemsets $I(G) \subseteq \L'$.
Further observe that for two frequent subgraphs $G$ and $H$,
$G \subseteq H$ if and only if $f(G) \subseteq f(H)$.
It follows that a graph $G$ and an itemset $I$ must have the same supports in
$D_{\MaxFFS{\Graphs}}$ and $D_{\MaxFIS}$, respectively. The maximality then follows from the
subset-property we observed.
\end{proof}

Note that the reduction simplifies when $\phi_{\MaxFFS{\Graphs}} \equiv 1$,
i.e., when we consider the reduction from frequency-based problem $\MaxFS{\Graphs}$
to the \ffpp $\MaxFFIS$.
Then the mapping $f$ stays the same and $\phi_{\MaxFFIS}$ only needs to check
condition~(2). We believe that many algorithms for mining itemsets can be
augmented with this choice of $\phi_{\MaxFFIS}$ function to mine graph
patterns as we will discuss further in Section~\ref{sec:algor-exper}.

Observe that while condition~(2) looks rather technical, it can be easily
implemented using a graph traversal. Additionally, when computing the union of
two feasible patterns, an algorithm only needs to check if both patterns share
any label.

Note also that the reduction above works as well for directed graphs (we just
need to distinguish between edge labels $(\labelop{u}, \labelop{v})$
and $(\labelop{v}, \labelop{u})$). This immediately gives us the
following lemma.
\begin{lemma}
  \label{thm:MaxFFSDirG_MaxFFIS}
	There exists a maximality-preserving reduction from $\MaxFFS{\DirG}$
	to $\MaxFFIS$.
\end{lemma}

\paragraph*{From sequences to feasible DAGs}
To finish the hierarchy of Figure~\ref{fig:hierarchy}, we need one more
reduction, from $\MaxFSQSn$ to $\MaxFFS{\DAG}$.

\begin{lemma}
  \label{thm:MaxFSQS_MaxFFSDAG}
  There exists a maximality-preserving reduction from $\MaxFSQSn$ to $\MaxFFS{\DAG}$.
\end{lemma}
\begin{proof}
Let $D_{\MaxFSQSn}$ be a database of sequences over
labels from $\L$, let $\tau$ be a support threshold, and let
$\phi_{\MaxFSQSn}$ be a feasibility function.
Recall that a sequence contains each label at most once.

\emph{Construction of $f$.}
For $\MaxFFS{\DAG}$ we use the same labels $\L$.
Consider a sequence $S \in \L^r$ of length $r$ such that $S_i \neq S_j$ for all $i \neq j$.
This sequence is mapped to the graph $G(S)$ with vertices
$V(S) = \{S_1, \dots, S_r\}$, where each vertex $S_i$ is labelled by
$\labelop{S_i}$.
The graph contains the edges
\begin{align*}
E(S) = \{ (S_i, S_j) : i \in \{1, \dots, k-1 \}, j > i\}.
\end{align*}

Given a DAG $p \in \patterns{\MaxFFS{\DAG}}$, we set
$\phi_{\MaxFFS{\DAG}}(p) = 1$ iff
$f^{-1}(p)$ exists and $\phi_{\MaxFSQSn}(f^{-1}(p)) = 1$.

\emph{Maximality-preserving.}
Note that Properties~\ref{item:patternSubsets2}--\ref{item:efficientInverse2}
of maximality-preserving reductions for $f$ are
trivially satisfied. We prove Property~\ref{item:maxPreserving2}.

Let $S$ be sequence from $\Max{D_{\MaxFSQSn}, \tau, \phi_{\MaxFSQSn}}$
of length~$r$.
We show that $G := f(S) \in \Max{ D_{\MaxFFS{\DAG}}, \tau, \phi_{\MaxFFS{\DAG}}}$.
By construction of $D_{\MaxFFS{\DAG}}$ and due to Property~\ref{item:subsetPreserving2},
$G$ is frequent in $D_{\MaxFFS{\DAG}}$.
We need to argue that $G$ is also maximal; we do this by contradiction.
Suppose there exists a feasible graph $H$ such that $G \subset H$.
Observe that adding any edge to $G$ would introduce a cycle.
Hence, $H$ must contain more vertices than $G$.
Since $H$ is also feasible, it corresponds to a sequence $S' = f^{-1}(H)$ of
length at least $r+1$. By Property~\ref{item:subsetPreserving2}, $S'$ is
frequent and $S \sqsubset S'$. This contradicts the maximality of $S$.

Consider any maximal feasible frequent DAG $G = (V,E)$ in $D_{\MaxFFS{\DAG}}$.
Since $G$ is feasible, let $S = f^{-1}(G)$.
Then the sequence $S = \langle v_1, \dots, v_r \rangle$ must be frequent in $D_{\MaxFSQSn}$
by the choice of $f$ and the construction of $D_{\MaxFFS{\DAG}}$.
Additionally, $S$ must be maximal.
Assume it is not. Then there exists a maximal sequence $T$ with $S \sqsubset T$.
By the argument of the previous paragraph, the graph $H = f(T)$ is maximal and frequent.
But then we also have $G = f(S) \sqsubset f(T) = H$, which contradicts the maximality of $G$.
\end{proof}

\section{Algorithms and Experiments}
\label{sec:algor-exper}

In this section, we discuss the practical consequences of our reductions
and show that the reductions can be used to develop efficient real-world
algorithms.

\subsection{Reductions as Algorithms} \label{sec:reduct-as-algor}

In addition to providing us the theoretical understanding of the relationships
between the problems, the reductions also provide us a direct way to solve a
maximal frequent pattern mining problem in one domain by using a solver from the
other domain. As an example, consider the reduction from the
\emph{frequency-based} problem $\MaxFS{\Graphs}$ to the \ffpp{}
$\MaxFFIS$ (Lemma~\ref{Thm:MaxFSGraphsFFIS}) and let $D_G$ be the graph
database for an instance of $\MaxFS{\Graphs}$ and $D_T$ be the transaction
database built by the reduction.

The mapping of patterns $f$ is straight forward, as we only need to generate a
transaction for each graph, and an item for each unique edge label. The crux of
the reduction lies in the feasibility function $\phi$: it has to ensure that the
returned frequent itemsets correspond to \emph{connected} frequent subgraphs in
the original problem. As the feasible frequent itemsets are a strict subset of
all of the frequent subsets,\!\footnote{Note, however, that the feasible
	\emph{maximal} itemsets are not necessarily a subset of all maximal
		itemsets.} we could simply prune out the results at the very end. A
		na\"ive algoritm for solving $\MaxFS{\Graphs}$ could then work as follows:
		(1) build $D_T$ following Lemma~\ref{Thm:MaxFSGraphsFFIS}; (2) compute
		all frequent itemsets from $D_T$; (3) prune out the non-feasible
		frequent itemsets; (4) prune out the non-maximal feasible frequent
		itemsets. 

More efficient implementations are possible, however. In particular, we can add
the feasibility constraint in the mining process, thus reducing the number of
candidates to consider in each iteration. The connectedness constraint is not
monotone, though: it is possible that two itemsets $A$ and $B$ do not correspond
to connected subgraphs, while their union does (e.g., $A=\{(a, b), (c, d)\}$ and
		$B=\{(b,c), (d, e)\}$). On the other hand, if $C$ is a feasible
(connected) frequent itemset in $D_T$, then it can be split into subsets of any
size that are frequent and feasible. This means that we can prune all
infeasible itemsets at the same time when we prune away all infrequent
itemsets. In other words, we can in fact work with \emph{less} candidates (or at
		least with no more) than if we would be doing standard frequent itemset
mining.

The final question in our example is how to implement the feasibility check
efficiently. Let us abuse the notation slightly and denote by $\labelop{A}$ the
set of unique (vertex) labels in an itemset $A$, that is $\labelop{A} = \{l :
\text{edge }(l, \cdot)\text{ or }(\cdot, l)\text{ is an item in } A\}$. Then
$A\cup B$ is a connected (i.e., feasible) itemset if and only if
$\labelop{A}\cap\labelop{B}\neq\emptyset$ and both $A$ and $B$ are connected
(i.e., feasible).  Hence, if we store the sets
$\labelop{A}$ together with the candidate itemsets, we only need to test the
disjointness of these two sets to test the feasibility of $A\cup B$.

The above example should make clear that the reductions we present in this
paper can yield practical algorithms, and it is not too hard to see that
similarly efficient algorithm can be designed following the reduction of
Lemma~\ref{thm:MaxFSQS_MaxFFSDAG}. However, note that in this reduction it would
not be a good idea to add single edges during the candidate generation; an
efficient implementation would ensure that whole nodes with edges to all over vertices are
added. This ensures that the preimage of the reduction exists at all times and that fewer
infeasible candidates are generated.

To further validate our approach, we present some experimental evaluation of the above
algorithm in the next subsection. Before that, let us however discuss a bit on
the general approaches for using the maximality-preserving reductions.

The first observation is that the type of the feasibility constraint obviously
has a big impact on the efficiency of the final algorithm. The study of
constrained frequent pattern mining is well established (see,
e.g.,~\cite{han07frequent} or Section~\ref{Sec:complexity}), and that
research gives characterizations of constraints that can be implemented
efficiently in standard algorithms. Similarly, the constraint-programming
algorithms for data analysis can often be easily adapted for the feasibility
constraints used in frequency-based reductions.

The second observation concerns the number of (non-maximal) frequent itemsets.
Our reductions are only guaranteed to preserve the maximality, and can, in
principle, yield an exponentially larger number of non-maximal frequent (and
feasible) itemsets. This would, naturally, make it practically
infeasible to use the reductions together with standard frequent pattern mining
algorithms. There are a few possible solutions to this. First, many reductions do
not grow the number of feasible frequent patterns. This is, for example, the
case with the reductions in Lemmas~\ref{Thm:MaxFIS}, \ref{Thm:ReductionMaxSQS},
and~\ref{Thm:MaxFSGraphsFFIS}.
Second, a clever implementation of a reduction would only generate candidates
which may be generated by the mapping from the reduction. This can dramatically
decrease the number of possible candidates. In fact, if the implementation
manages not to generate any candidates which have no preimage under the mapping
from the reduction, then the number of possible candidates will not increase at
all. We believe that this is possible for all reductions we present in this paper.
Third, the maximal frequent patterns can also
be found by first finding all the maximal frequent and minimal infrequent
patterns~\cite{gunopulos97data}. Unfortunately for this approach, we do not
yet know the behaviour of minimal infrequent patterns under our reductions.
We leave further studies in this for future work.

\subsection{Experimental Evaluation} \label{sec:exper-eval}

For the experimental evaluation, we implemented the reduction from
$\MaxFS{\Graphs}$ to $\MaxFIS$ (Lemma~\ref{Thm:MaxFSGraphsFFIS}) in a custom
version of the Apriori algorithm~\cite{agrawal96fast}. The constraint on the
feasible patterns was straight forward to implement, as discussed above.\footnote{The code and sample data are available from \codeURL.} 

We tested our approach on a discussion forum data from the
Stack\-Exchange forums.\footnote{\url{
https://archive.org/details/stackexchange}} The data contains 161 different
question-answering forums (we excluded the meta-forums). We concentrated on the
most recent year's activity, and constructed one graph for each forum where the
users are the vertices and there is an edge between two users if one has
answered or commented to the other's question or answer. The vertices are
labelled uniquely using the global user-id. The data has $1\,627\,946$ different
users, and in total $8\,264\,675$ uniquely-labelled edges. Hence, the dataset
does not pose a significant problem for frequent itemset mining algorithms.

We wanted to study the effects the constraint has for the number of candidates.
Recall that the constraint is used to enforce that we find only connected
subgraphs. In Figure~\ref{fig:candidates_pruned}, we show the number of frequent
itemsets and feasible frequent itemsets of different sizes with minimum
frequency $3$. 

\begin{figure}
  \centering
  \begin{tikzpicture}
    \begin{semilogyaxis}[
      xlabel={Itemset size},
      ylabel={Number of itemsets},
      legend entries={All frequent,Only feasible},
      legend pos=south west]
      \addplot table {
        x y
        1 679
        2 613
        3 1288
        4 3545
        5 8985
        6 18842
        7 31961
        8 43805
        9 48630
        10 43759
        11 31824
        12 18564
        13 8568
        14 3060
        15 816
        16 153
        17 18
        18 1
      };
      \addplot table {
        x y
        1 679
        2 277
        3 298
        4 457
        5 854
        6 1702
        7 3163
        8 4950
        9 6038
        10 5454
        11 3552
        12 1652
        13 540
        14 119
        15 16
        16 1
      };
    \end{semilogyaxis}
  \end{tikzpicture}
  \caption{The number of frequent itemsets and feasible frequent itemsets when
	  solving the $\MaxFS{\Graphs}$ problem using $\MaxFIS$ algorithms.
	  The $y$-axis is in logarithmic scale.}
  \label{fig:candidates_pruned}
\end{figure}
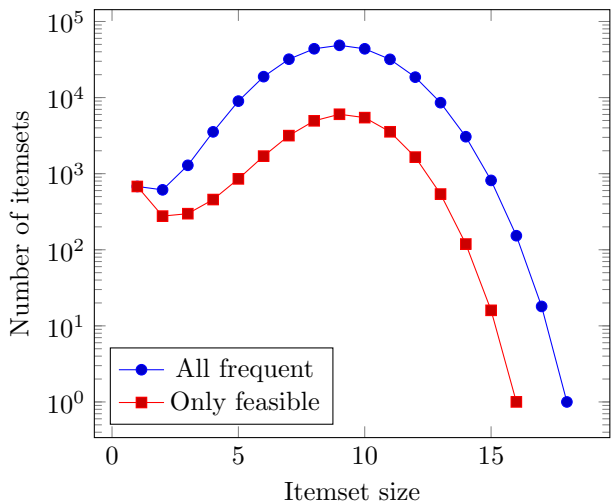

As can be seen from Figure~\ref{fig:candidates_pruned}, the total number of
frequent itemsets is approximately ten times the number of feasible candidates,
		 indicating that the feasibility constraint allows us to prune
		 significant amounts of candidates (there are no feasible candidates of
				 size 17 or 18). In total, the data has $265\,111$ frequent
		 itemsets, of which $29\,752$ were feasible and $549$ were maximal
		 feasible itemsets.

%% BEGIN can be removed if we get short on space
The number of maximal frequent itemsets and maximal feasible frequent itemsets
with respect to different minimum thresholds is presented in
Figure~\ref{fig:maximal_vs_threshold}. We can see that their numbers are mostly
aligned, with the number of maximal itemsets dropping almost exponentially as
the minimum threshold increases. No pattern has support higher than $9$.

\begin{figure}
  \centering
  \begin{tikzpicture}
    \begin{axis}[
      xlabel={Minimum frequency},
      ylabel={Number of maximal itemsets},
      legend entries={All itemsets,Only feasible},
      ]
      \addplot table {
        x y
        3 499
        4 79
        5 12
        6 3
        7 2
        8 2
        9 1
      };
      \addplot table {
        x y
        3 549
        4 79
        5 12
        6 3
        7 2
        8 2
        9 1
      };
    \end{axis}
  \end{tikzpicture}
  \caption{The number of maximal feasible frequent itemsets with different
	  minimum support thresholds.}
  \label{fig:maximal_vs_threshold}
\end{figure}
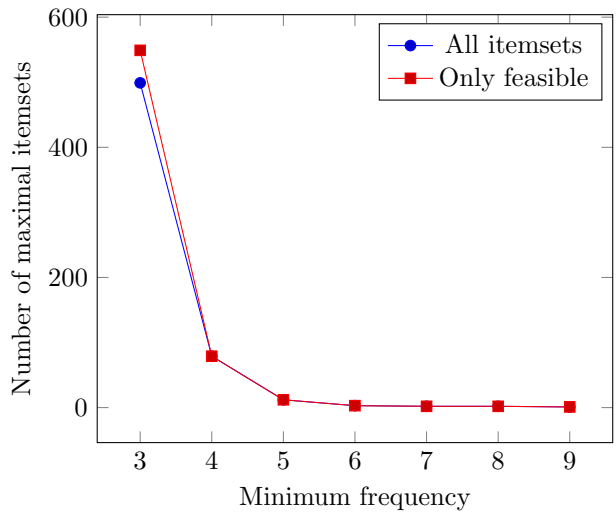

%% END can be removed if we get short on space

%%% Local Variables:
%%% mode: latex
%%% TeX-master: "icmd"
%%% End:

\section{Conclusion}
\label{Sec:conclusion}

We showed that when considering a generalized version of frequency-based
problems, \ffpp, the computational hardness of many frequency-based
problems collapses.
Hence, our reductions provide a unifying framework for the existing
computational hardness results of fundamental data mining problems.
Additionally, our reductions give a formal explanation why algorithms similar
to the Apriori algorithm can be used for such a wide range of
problems by only slightly adjusting the candidate generation.

In the future it will be interesting to study the computational complexity
of frequency-based problems in which labels can appear multiple times.
A daunting question is whether the following two problems exhibit the same hardness:
Mining subsequences without the restriction that each label appears only once,
and mining graphs with possibly multiple vertices of the same label.

The reductions we provide hint that many practical algorithms for frequency-based problems
can be augmented to solve more complicated problems.
We provided such an example in Section~\ref{sec:algor-exper}.
It will be interesting to see if our insights can lead to more efficient algorithms 
for the problems we considered or to algorithms which can solve a wider range of problems.

\bibliography{main}{}
\bibliographystyle{plain}

\end{document}